\def\llncs{0}
\def\fullpage{1}
\def\anonymous{0}
\def\authnote{1}
\def\notxfont{1}
\def\submission{0}

\ifnum\submission=1
\def\llncs{1}
\fi

\ifnum\llncs=1
	\documentclass[envcountsect,a4paper,runningheads,11pt]{llncs}
  \usepackage[margin=1in]{geometry}

\else
	\documentclass[11pt]{article}
			\ifnum\fullpage=1
		\usepackage{fullpage}
		\fi
\fi

\usepackage[%
  colorlinks=true,
  citecolor=blue,
  pagebackref=true
]{hyperref}

\usepackage{amsmath, amssymb, mathtools,amscd}

\usepackage{amsthm}

\usepackage{arydshln} 
\usepackage{url}
\usepackage{ifthen}
\usepackage{bm}
\usepackage{multirow}
\usepackage[dvips]{graphicx}
\usepackage[usenames]{color}
\usepackage{xcolor,colortbl} 
\usepackage{threeparttable}
\usepackage{comment}
\usepackage{paralist,verbatim}
\usepackage{cases}
\usepackage{booktabs}
\usepackage{braket}
\usepackage{cancel} 
\usepackage{ascmac} 
\usepackage{framed}
\usepackage{authblk}
\usepackage{pifont}
\usepackage{qcircuit}
\definecolor{darkblue}{rgb}{0,0,0.6}
\definecolor{darkgreen}{rgb}{0,0.5,0}
\definecolor{maroon}{rgb}{0.5,0.1,0.1}
\definecolor{dpurple}{rgb}{0.2,0,0.65}

\usepackage[capitalise,noabbrev]{cleveref}
\usepackage[final]{microtype}
\usepackage[absolute]{textpos}
\usepackage{everypage}
\DeclareMathAlphabet{\mathpzc}{OT1}{pzc}{m}{it}

\usepackage{algorithmic}
\usepackage{algorithm}
\usepackage{here}

\newtheoremstyle{thicktheorem}%
{\topsep}
{\topsep}
{\itshape}{}%
{\bfseries}%
{.}
{ }%
{\thmname{#1}\thmnumber{ #2}%
		\thmnote{ (#3)}%
}

\newtheoremstyle{remark}
{\topsep}
{\topsep}
	{}
	{}
	{}
	{.}
	{ }
	{\textit{\thmname{#1}}\thmnumber{ #2}
			\thmnote{ (#3)}%
	}

\ifnum\llncs=0
	\theoremstyle{thicktheorem}
	\newtheorem{theorem}{Theorem}[section]
	\newtheorem{lemma}[theorem]{Lemma}
	\newtheorem{corollary}[theorem]{Corollary}
	
	\newtheorem{definition}[theorem]{Definition}

	\theoremstyle{remark}

\else
\fi

\Crefname{MyClaim}{Claim}{Claims}

	\crefname{theorem}{Theorem}{Theorems}
	\crefname{assumption}{Assumption}{Assumptions}
	\crefname{construction}{Construction}{Constructions}
	\crefname{corollary}{Corollary}{Corollaries}
	\crefname{conjecture}{Conjecture}{Conjectures}
	\crefname{definition}{Definition}{Definitions}
	\crefname{exmaple}{Example}{Examples}
	\crefname{experiment}{Experiment}{Experiments}
	\crefname{counterexample}{Counterexample}{Counterexamples}
	\crefname{lemma}{Lemma}{Lemmata}
	\crefname{observation}{Observation}{Observations}
	\crefname{proposition}{Proposition}{Propositions}
	\crefname{remark}{Remark}{Remarks}
	\crefname{claim}{Claim}{Claims}
	\crefname{fact}{Fact}{Facts}
	\crefname{note}{Note}{Notes}

\ifnum\llncs=1
 \crefname{appendix}{App.}{Appendices}
 \crefname{section}{Sec.}{Sections}
\else
\fi

\ifnum\llncs=1
\renewcommand*{\backref}[1]{}
\else
	\renewcommand*{\backref}[1]{(Cited on page~#1.)}
	\ifnum\notxfont=1
	\else
		\usepackage{newtxtext}
	\fi
\fi
\ifnum\authnote=0  
\newcommand{\mor}[1]{}
\newcommand{\minki}[1]{}
\newcommand{\takashi}[1]{}
\newcommand{\michael}[1]{}

\else
\newcommand{\mor}[1]{$\ll$\textsf{\color{red} Tomoyuki: { #1}}$\gg$}
\newcommand{\takashi}[1]{$\ll$\textsf{\color{orange} Takashi: { #1}}$\gg$}
\newcommand{\minki}[1]{$\ll$\textsf{\color{darkgreen} Minki: { #1}}$\gg$}
\newcommand{\michael}[1]{$\ll$\textsf{\color{darkgreen} Michael: { #1}}$\gg$}

\fi

\newcommand{\Tr}{\mathrm{Tr}}

\newcommand{\StateGen}{\mathsf{StateGen}}

\newcommand{\abs}[1]{|#1|}

\newcommand{\cA}{\mathcal{A}}

\newcommand{\cH}{\mathcal{H}}

\newcommand{\cX}{\mathcal{X}}
\newcommand{\cY}{\mathcal{Y}}

\def\makeuppercase#1{
\expandafter\newcommand\csname tl#1\endcsname{\widetilde{#1}}
}

\def\makelowercase#1{
\expandafter\newcommand\csname tl#1\endcsname{\widetilde{#1}}
}

\newcommand{\regC}{\mathbf{C}}
\newcommand{\regR}{\mathbf{R}}
\newcommand{\regZ}{\mathbf{Z}}

\newcommand{\regB}{\mathbf{B}}
\newcommand{\regA}{\mathbf{A}}

\newcommand{\regG}{\mathbf{G}}

\newcommand{\eps}{\epsilon}

\newcommand{\secp}{\lambda}

\newcommand{\A}{\entity{A}}

\newcommand*{\algo}[1]{\ensuremath{\mathsf{#1}}}

\newcommand*{\entity}[1]{\mathcal{#1}}

\newenvironment{boxfig}[2]{\begin{figure}[#1]\fbox{\begin{minipage}{0.97\linewidth}
                        \vspace{0.2em}
                        \makebox[0.025\linewidth]{}
                        \begin{minipage}{0.95\linewidth}
            {{
                        #2 }}
                        \end{minipage}
                        \vspace{0.2em}
                        \end{minipage}}}{\end{figure}}

\newcommand{\bit}{\{0,1\}}

\newcommand{\KeyGen}{\algo{KeyGen}}

\newcommand{\Ver}{\algo{Ver}}

\newcommand{\TD}{\algo{TD}}

\newcommand{\negl}{{\mathsf{negl}}}

\newcommand{\poly}{{\mathsf{poly}}}

\makeatletter
\DeclareRobustCommand
  \myvdots{\vbox{\baselineskip4\p@ \lineskiplimit\z@
    \hbox{.}\hbox{.}\hbox{.}}}
\makeatother

\newcommand{\proj}[1]{\ket{#1}\!\!\bra{#1}}

\allowdisplaybreaks[4]

\renewcommand{\paragraph}[1]{\medskip\noindent\textbf{#1}}
\newcommand{\email}[1]{\href{mailto:#1}{#1}}
\numberwithin{equation}{section}

\title{Exponential Quantum One-Wayness and EFI Pairs}

\ifnum\anonymous=1
\ifnum\llncs=1
\author{\empty}\institute{\empty}
\else
\author{}
\fi
\else
\ifnum\llncs=1
\author{
Giulio Malavolta\inst{1,2}
\and Tomoyuki Morimae\inst{3}
 \and Michael Walter\inst{4}
 \and Takashi Yamakawa\inst{5,6,3}
}
\institute{
Bocconi University, Milan, Italy
\and Max Planck Institute for Security and Privacy, Bochum, Germany
\and	Yukawa Institute for Theoretical Physics, Kyoto University, Kyoto, Japan
\and Faculty of Computer Science, Ruhr University Bochum, Bochum, Germany
 \and NTT Social Informatics Laboratories, Tokyo, Japan
 \and NTT Research Center for Theoretical Quantum Information, Atsugi, Japan
}
\else
\author[1,2]{Giulio Malavolta}
\author[3]{Tomoyuki Morimae}
\author[4]{Michael Walter}
\author[5,6,3]{Takashi Yamakawa}
\affil[1]{{\small Bocconi University, Milan, Italy}\authorcr{\small \email{giulio.malavolta@hotmail.it}}}
\affil[2]{{\small Max Planck Institute for Security and Privacy, Bochum, Germany}}
\affil[3]{{\small Yukawa Institute for Theoretical Physics, Kyoto University, Kyoto, Japan}\authorcr{\small \email{tomoyuki.morimae@yukawa.kyoto-u.ac.jp}}}
\affil[4]{{\small Faculty of Computer Science, Ruhr University Bochum, Bochum, Germany}\authorcr{\small \email{michael.walter@rub.de}}}
\affil[5]{{\small NTT Social Informatics Laboratories, Tokyo, Japan}\authorcr{\small \email{takashi.yamakawa@ntt.com}}}
\affil[6]{{\small NTT Research Center for Theoretical Quantum Information, Atsugi, Japan}}

\fi 
\fi

\date{}

\begin{document}

\maketitle
\begin{abstract}
In classical cryptography, one-way functions are widely considered to be the minimal computational assumption.
However, when taking quantum information into account, the situation is more nuanced.
There are currently two major candidates for the minimal assumption:
the \emph{search} quantum generalization of one-way functions are one-way state generators (OWSG), whereas the \emph{decisional} variant are EFI pairs.
A well-known open problem in quantum cryptography is to understand how these two primitives are related.
A recent breakthrough result of Khurana and Tomer (STOC'24) shows that OWSGs imply EFI pairs, for the restricted case of pure states.

In this work, we make progress towards understanding the general case.
To this end, we define the notion of \emph{inefficiently-verifiable one-way state generators}~(IV-OWSGs), where the verification algorithm is not required to be efficient, and show that these are precisely equivalent to EFI pairs, with an exponential loss in the reduction.
Significantly, this equivalence holds also for \emph{mixed} states.
Thus our work establishes the following relations among these fundamental primitives of quantum cryptography:
\[
 \text{(mixed) OWSGs}  \implies  \text{(mixed) IV-OWSGs} \equiv_\text{exp} \text{EFI pairs},
\]
where $\equiv_\text{exp}$ denotes equivalence up to exponential security of the primitives.
\end{abstract}

\section{Introduction}
The existence of one-way functions (OWFs) is widely regarded as the minimal assumption in classical cryptography.
This is because almost all primitives imply OWFs and furthermore OWFs are in fact equivalent to many foundational primitives, such as secret-key encryption (SKE), commitments, zero-knowledge, pseudorandom generators (PRGs), pseudorandom functions (PRFs), and digital signatures~\cite{STOC:LubRac86,FOCS:ImpLub89,STOC:ImpLevLub89}.
However, recent works have suggested that OWFs may \emph{not} be the minimal assumption when bringing quantum information into the picture.
Instead, several candidate ``minimal'' primitives have been proposed that are potentially weaker than OWFs~\cite{Kre21,STOC23:KreQiaSinTal,cryptoeprint:2023/1602}, yet still enable many useful applications, such as private-key quantum money, SKE, commitments, multiparty computations, and digital signatures~\cite{C:JiLiuSon18,C:MorYam22,C:AnaQiaYue22,ITCS:BCQ23,AC:Yan22}.

In this work, we consider two major candidates that have emerged in this recent line of research:
\emph{one-way state generators (OWSGs)} and \emph{EFI pairs}.
OWSGs were introduced in~\cite{C:MorYam22,cryptoeprint:2022/1336} as a \emph{search} quantum generalization of OWFs.
Formally, a OWSG consists of a triple~$(\KeyGen,\StateGen,\Ver)$ of quantum polynomial-time (QPT) algorithms,
where~$\KeyGen(1^\secp)\to k$ is the key generation algorithm (keys are classical),
$\StateGen(k)\to \phi_k$ takes a key as input and generates a quantum state, and
$\Ver(k', \phi) \to \{\top,\bot\}$ is a verification algorithm that takes a bit string $k'$ and a quantum state $\phi$ as input.
The security of a OWSG requires that no QPT adversary~$\mathcal A$ can find a ``preimage'' of $\phi_k$ with non-negligible probability, that is,
\begin{align*}
    \Pr\left[\top\gets\Ver(k',\phi_k) \;:\; k\gets\KeyGen(1^\secp),\phi_k\gets\StateGen(k),k'\gets\cA(1^\secp,\phi_k^{\otimes t})\right]
\approx 0
\end{align*}
for any polynomial number~$t=t(\lambda)$ of copies.
This can be considered as a quantum analogue of the one-wayness of OWFs.

On the other hand, EFI pairs~\cite{ITCS:BCQ23} are a \emph{decisional} quantum generalization of OWFs.
Formally, an EFI pair consists of a QPT algorithm that on input~$1^\secp$ generates two (mixed) quantum states~$\xi_0$ and~$\xi_1$ which are statistically far but computationally indistinguishable.
As such, EFI pairs are a quantum generalization of EFIDs, which are pairs of efficiently samplable classical distributions that are statistically distinguishable but computationally indistinguishable, a primitive that is equivalent to OWFs~\cite{Gol90}.

Given that both of these quantum primitives generalize the same object in classical cryptography, it is natural to ask about their relation.
A recent breakthrough result~\cite{cryptoeprint:2023/1620} shows that OWSGs imply EFI pairs provided the outputs of~$\StateGen$ are \emph{pure} states.
The general case where the outputs of~$\StateGen$ are \emph{mixed states}, as well as the reverse direction of the implication, remain~open.

\subsection{Our Results}

In this work, we make progress towards understanding the relation between OWSGs and EFI pairs.
First, we define a weaker notion of OWSGs, wherein the verification algorithm~$\Ver$ is not required to be efficient.
We refer to this primitive as \emph{inefficiently-verifiable one-way state generators (IV-OWSGs)}.
Then, as the main technical contribution of our work, we show that IV-OWSGs and EFI pairs are equivalent, but with an exponential loss in the reduction.
That is, our work establishes the following relations amongst these fundamental primitives in quantum cryptography:
\[
 \text{(mixed) OWSGs}  \implies  \text{(mixed) IV-OWSGs}
\equiv_\text{exp} \text{EFI pairs}
\]
where
$\equiv_\text{exp}$
means equivalence with an exponential loss for the implication from the left to right (but only with a polynomial loss for the other direction)
and we write ``(mixed)'' to stress that all our results hold in the general case that
the $\StateGen$ of OWSGs returns mixed states.
The first implication is clear.
That EFI pairs imply (mixed) IV-OWSGs follows using known facts in quantum cryptography (\cref{thm1}), and the reduction does not incur any loss.
Our main technical contribution is to prove that 
IV-OWSGs imply EFI pairs with an exponential security loss (\cref{thm:EFI_from_IV-OWSG}).
Crucially, all our results hold in the general setting where the $\StateGen$ of OWSGs is allowed to generate mixed states.
In contrast, the recent breakthrough \cite{cryptoeprint:2023/1620} only considered \emph{pure} states (but their reduction has a polynomial loss, so the results are incomparable).
Our proof follows a different route, and makes crucial use of Aaronson's shadow tomography algorithm~\cite{Shadow2}.

Using known implications from the literature~\cite{ITCS:BCQ23}, a consequence of our work is that (mixed) OWSGs with exponential security imply a number of primitives in cryptography, such as non-interactive commitments, quantum computational zero knowledge, oblivious transfer, and general multiparty computation.
At a more conceptual level, our work also sheds some light on the relation between OWSGs and EFI pairs:
For instance, if one were able to show the outstanding implication that EFI pairs $\implies$ OWSGs (with a polynomial reduction), one would automatically obtain a generic conversion that turns any exponentially secure OWSGs with \emph{inefficient} verification into one with \emph{efficient} verification, which would perhaps be surprising.

\subsection{Proof Outline}
We give here a brief overview of the proof of our main technical contribution, namely that IV-OWSGs imply EFI pairs.
It is well-known that EFI pairs are equivalent to canonical quantum bit commitments~\cite{AC:Yan22}, and therefore our goal is to construct commitments from IV-OWSGs.
Consider an arbitrary IV-OWSG $(\KeyGen,\StateGen,\Ver)$.
We construct a non-interactive bit commitment scheme in the canonical form of~\cite{AC:Yan22} as follows.
A commitment to $0$ is a state
\begin{align*}
\ket{\psi_{0}}_{\regR,\regC}\coloneqq
\sum_{k}
\sum_{h\in\cH}
\sqrt{\frac{\Pr[k\gets\KeyGen(1^\secp)]}{|\cH|}}
\ket{k,\mathrm{junk}_k,h}_{\regC_1}\ket{h,h(k)}_{\regR_1}\ket{\Phi_k^{\otimes t}}_{\regC_2,\regR_2}\ket{0...0}_{\regR_3}
\end{align*}
whereas a commitment to $1$ is a state
\begin{align*}
\ket{\psi_1}_{\regR,\regC}\coloneqq
\sum_{k}
\sum_{h\in\cH}
\sqrt{\frac{\Pr[k\gets\KeyGen(1^\secp)]}{|\cH|}}
\ket{k,\mathrm{jun}k_k,h}_{\regC_1}\ket{h,h(k)}_{\regR_1}\ket{\Phi_k^{\otimes t}}_{\regC_2,\regR_2}\ket{k}_{\regR_3}
\end{align*}
where $t=t(\secp)$ is a certain polynomial specified later, $\mathcal{H}$ is a family of pairwise-independent hash functions,
$\ket{\mathrm{junk}_k}$ is the state of the non-output registers in a unitary realization of the $\KeyGen$ algorithm,%
\footnote{Without loss of generality, the $\KeyGen$ algorithm takes the following form:
apply a QPT unitary to generate a superposition~$\sum_k \sqrt{\Pr[k\gets\KeyGen(1^\secp)]}\ket{k}\ket{\mathrm{junk}_k}$,
measure the first register, and output the measurement result.
}
and $\ket{\Phi_k}$ is a purification of the output of the $\StateGen(k)$ algorithm.\footnote{
Without loss of generality, $\StateGen$ takes the following form: on input $k$, apply a QPT unitary~$U_k$ on~$\ket{0...0}$ to generate a pure state~$\ket{\Phi_k}_{\regA,\regB}=U_k|0...0\rangle$ and output the first register~$\regA$, which is in state $\phi_k = \Tr_{\regB}(\proj{\Phi_k})$.
Then the~$\regA$ registers of~$\ket{\Phi_k^{\otimes t}}$ make up~$\regR_2$, while the~$\regB$ registers make up~$\regC_2$.}
One should think of~$\regC\coloneqq(\regC_1,\regC_2)$ as the commitment register and of~$\regR\coloneqq(\regR_1,\regR_2,\regR_3)$ as the reveal register.
It is clear that both $\ket{\psi_0}$ and $\ket{\psi_1}$ can be efficiently generated.

Computational binding is shown with a reduction to the exponential one-wayness of IV-OWSG:
Assuming that there exists an efficient algorithm that converts $\ket{\psi_{0}}_{\regR,\regC}$ into $\ket{\psi_{1}}_{\regR,\regC}$ by acting only on $\regR$, we can obtain an algorithm that computes $k$ from $h$, $h(k)$, and $\phi_k^{\otimes t} = \Tr_{\regC_2}(\proj{\Phi_k^{\otimes t}})$. By randomly guessing $h(k)$, we can use it to break the security of IV-OWSG with an exponential security loss of $2^{|h(k)|}$ where $|h(k)|$ is the output length of $h$.

To show statistical hiding, it suffices to construct an (inefficient) unitary on $\regR$ that turns
$\ket{\psi_{0}}_{\regR,\regC}$ into $\ket{\psi_{1}}_{\regR,\regC}$ with a sufficiently good approximation.
To this end consider the following procedure:
\begin{enumerate}
    \item Apply shadow tomography~\cite{Shadow2} to list all $k'$ that are accepted by $\Ver(\cdot,\phi_k)$ with a sufficiently large probability.
    \item In this list, find $k^*$ such that $h(k)=h(k^*)$. If there is a single such $k^*$, output it. Otherwise, output $\bot$.
\end{enumerate}
Intuitively, this algorithm outputs the true key $k$ with probability at least $1/2-\negl(\secp)$,
because the list obtained by the shadow tomography contains~$k$ except for a negligible probability
and because we set the output length of $h$ appropriately so that, with probability at least~$1/2$, $h^{-1}(h(k))$ intersects the list in a single element.
By running this algorithm coherently, we get an (inefficient) unitary on $\regR$ that maps $\ket{\psi_0}_{\regR,\regC}$ close to $\ket{\psi_1}_{\regR,\regC}$.
Thus Uhlmann's theorem implies that the trace distance between
$\Tr_{\regR}(\ket{\psi_0}_{\regR,\regC})$
and
$\Tr_{\regR}(\ket{\psi_1}_{\regR,\regC})$
is small.
This shows statistical hiding.

\subsection{Paper Outline}
\cref{sec:preliminaries} is for preliminaries where some notations and basic definitions are given.
In \cref{sec:def_IV-OWSGs}, we define the new notion of IV-OWSGs.
In \cref{sec:EFI_IVOWSG}, we show that EFI pairs imply IV-OWSGs, unconditionally and without any loss.
In \cref{sec:exp_IVOWSG_EFI}, we show that EFI pairs can be constructed from exponentially-secure IV-OWSGs, with an exponential loss.

\section{Preliminaries}
\label{sec:preliminaries}

\subsection{Basic Notations}
We use notation that is standard in quantum computing and cryptography.
We use $\secp$ as the security parameter.
The notation~$[n]$ refers to the set~$\{1,2,...,n\}$.
For any set $S$, $x\gets S$ means that an element $x$ is sampled uniformly at random from the set~$S$.
For an algorithm $A$, $y\gets A(x)$ means that the algorithm outputs~$y$ on input~$x$.
For a set $S$, $|S|$ denotes its cardinality.
We write~$\negl$ to denote a negligible function and $\poly$ to mean a polynomial function.
QPT stands for quantum polynomial-time.
Any binary (2-outcome) quantum measurement can be described by a POVM element~$E$, which is an operator such that both~$E$ and~$I-E$ are positive semidefinite.
Quantum registers are denoted by bold font (e.g., $\regA$ and $\regB$).
We write~$\Tr_\regB(\rho_{\regA,\regB})$ for the partial trace over the register~$\regB$ of the bipartite state~$\rho_{\regA,\regB}$.
The notation~$I_\regA$ denotes the identity operator on register~$\regA$.
For any two quantum states $\rho$ and $\sigma$, their fidelity is~$F(\rho,\sigma)\coloneqq(\Tr\sqrt{\sqrt{\sigma}\rho\sqrt{\sigma}})^2$ and their trace distance is~$\TD(\rho,\sigma)\coloneqq\frac{1}{2}\|\rho-\sigma\|_1$.

\subsection{Computational Model for Adversaries}
Throughout the paper, we treat adversaries as \emph{non-uniform} QPT machine with \emph{quantum} advice.
However, all of our results extend to the uniform setting in a straightforward manner.

\subsection{Pairwise-Independent Hash Family}
We recall the definition of pairwise independence.

\begin{definition}[Pairwise-Independent Hash Family]
A family $\mathcal{H} = \{h:\cX\to\cY\}$ of functions is called a \emph{pairwise-independent hash family}
if, for any two $x\neq x'\in \cX$ and any two $y,y'\in\cY$,
\[ \Pr_{h\gets\mathcal{H}}[h(x)=y\wedge h(x')=y']=\frac{1}{|\cY|^2}. \]
\end{definition}

\subsection{Shadow Tomography}

The shadow tomography problem asks to predict a large number of measurement outcomes from copies of a quantum state.
More formally:

\begin{definition}[Shadow Tomography Problem~{\cite[Problem 1]{Shadow2}}]
Given $t$ copies of an unknown $d$-dimensional (possibly mixed) quantum state~$\rho$, as well as known binary measurements given by POVM elements~$E_1,\dots,E_M$, 
output numbers $b_1, \dots , b_M \in [0, 1]$ such that
\[
|b_j - \Tr (E_j\rho)| \leq \varepsilon
\]
for all~$j$, with success probability at least $1-\omega$.
\end{definition}

The following theorem from Aaronson provides a bound on the number of copies of the state needed to solve this problem.

\begin{theorem}[{\cite[Theorem 2]{Shadow2}}]\label{lem:ST}
The shadow tomography problem is solvable with
\[
    t = \Tilde{O}\left(\frac{\log \frac1\omega}{\varepsilon^4}\cdot\log^4M\cdot\log d\right),
\]
where the $\Tilde{O}$ hides a $\poly(\log\log M,\log\log d,\log\frac{1}{\epsilon})$ factor.
\end{theorem}

\subsection{EFI Pairs and Quantum Bit Commitments}

EFI pairs are pairs of {\bf e}fficiently generatable quantum states that are statistically {\bf f}ar, yet computationally {\bf i}ndistinguishable.
They were introduced in \cite{ITCS:BCQ23}.

\begin{definition}[EFI Pair~\cite{ITCS:BCQ23}]
An \emph{EFI pair} is a family $\{\xi_{\secp,b}\}_{\secp\in\mathbb{N},\ b\in\bit}$ of (mixed) quantum states that satisfies the following conditions:
\begin{itemize}
\item (Efficiently Generatable)
There is a uniform QPT algorithm that generates~$\xi_{\secp,0}$ and~$\xi_{\secp,1}$ on input $1^\secp$.
\item (Statistically Far)
It holds that
$\TD(\xi_{\secp,0},\xi_{\secp,1})\ge 1/\poly(\secp)$.
\item (Computationally Indistinguishable)
 For any non-uniform QPT distinguisher $\cA$, there exists a negligible function~$\negl$ such that

\begin{align*}
 \left|\Pr[\A(1^\secp,\xi_{\secp,0})=1]-\Pr[\A(1^\secp,\xi_{\secp,1})=1]\right|\le \negl(\secp).
\end{align*}
\end{itemize}
\end{definition}

It is known that EFI pairs exist if and only if quantum bit commitments exist.
We define canonical quantum bit commitments following~\cite{AC:Yan22}.

\begin{definition}[Canonical Quantum Bit Commitments~\cite{AC:Yan22}]\label{def:canonical_com}
A \emph{canonical quantum bit commitment scheme} consists of a family $\{Q_0(\secp),Q_1(\secp)\}_{\secp\in \mathbb{N}}$ of uniform QPT unitaries.
Each acts on two registers $\regC$ (called the \emph{commitment register}) and $\regR$ (called the \emph{reveal register}).
In the rest of the paper, we often omit $\secp$ and simply write $Q_0$ and $Q_1$ to mean $Q_0(\secp)$ and $Q_1(\secp)$.
We define two properties that can be satisfied by canonical quantum bit commitments:
\begin{itemize}
\item (Hiding)
The scheme is \emph{computationally (resp.\ statistically) $\epsilon$-hiding} if for any non-uniform QPT (resp.\ for any unbounded) adversary~$\cA$, it holds that
\begin{align*}
\left| \Pr\left[ 1\gets\cA\Bigl(1^\secp,\Tr_{\regR}\bigl( (Q_0 \proj{0} Q_0^\dagger)_{\regC,\regR} \bigr) \Bigr) \right]
- \Pr\left[ 1\gets\cA\Bigl(1^\secp,\Tr_{\regR}\bigl( (Q_1 \proj{0} Q_1^\dagger)_{\regC,\regR} \bigr) \Bigr) \right] \right| \leq \epsilon(\secp).
\end{align*}
We say that the scheme is \emph{computationally (resp.\ statistically) hiding} if $\eps$ is negligible.

\item (Binding)
The scheme is \emph{computationally (resp.\ statistically) $\delta$-binding} if for any polynomial-size register~$\regZ$ and for any non-uniform QPT (resp.\ for any unbounded) unitary~$U_{\regR,\regZ}$, it holds that
\begin{align*}
\left\|
    \bigl( (\bra{0}Q_1^\dagger)_{\regC,\regR} \otimes I_{\regZ} \bigr)
    \bigl( I_{\regC}\otimes U_{\regR,\regZ} \bigr)
    \bigl( (Q_0\ket{0})_{\regC,\regR} \otimes I_{\regZ} \bigr)
\right\|\le \delta(\secp).
\end{align*}
We say that the scheme is \emph{computationally (resp.\ statistically) binding} if $\delta$ is negligible.
\end{itemize}
\end{definition}

Note that \emph{statistical} $\eps$-hiding can also be defined in terms of the trace distance, as follows:
\begin{align}\label{eq:stat hiding via TD}
    \TD\Bigl( \Tr_{\regR}\bigl( (Q_0 \proj{0} Q_0^\dagger)_{\regC,\regR} \bigr), \Tr_{\regR}\bigl( (Q_1 \proj{0} Q_1^\dagger)_{\regC,\regR} \bigr) \Bigr) \leq \epsilon(\secp).
\end{align}

The following lemma is an immediate consequence of Uhlmann’s theorem.

\begin{lemma}\label{lem:com_to_EFI}
EFI pairs exist if and only if there exists a canonical quantum bit commitment scheme that is computationally hiding and statistically $(1-1/\poly(\secp))$-binding.%
\footnote{A very recent work~\cite{parallel_repetition} shows that canonical quantum bit commitment schemes that satisfy computational hiding and \emph{computational} $(1-1/\poly(\secp))$-binding are sufficient for constructing EFI pairs. We do not need this result.}
\end{lemma}

It is known that binding and hiding of canonical quantum bit commitments can be traded for each other~\cite{EC:CreLegSal01,AC:Yan22,EC:HhaMorYam23,STOC:GJMZ23}.
In particular, the following result was shown in~\cite{EC:HhaMorYam23} (see also~\cite{Combiner}).

\begin{lemma}[Flavor Conversion for Quantum Bit Commitments~\cite{EC:HhaMorYam23}]\label{lem:flavor_conversion}
If there exists a canonical quantum bit commitment scheme that is statistically $\epsilon$-hiding and computationally binding, then there also exists one that is computationally hiding and statistically $\sqrt{\epsilon}$-binding.
\end{lemma}

Combining \cref{lem:com_to_EFI,lem:flavor_conversion}, we obtain the following corollary.

\begin{corollary}\label{cor:com_to_EFI}
If there exists a canonical quantum bit commitment scheme that is statistically $(1-1/\poly(\secp))$-hiding and computationally binding, then EFI pairs exist.
\end{corollary}

\section{Inefficiently-Verifiable One-Way State Generators (IV-OWSGs)}
\label{sec:def_IV-OWSGs}

In this section, we define IV-OWSGs and the notion of exponential security.
We first recall the definition of OWSGs.%
\footnote{The original OWSGs introduced in \cite{C:MorYam22} were defined to have pure state outputs, but this was later generalized to allow for mixed state outputs~\cite{cryptoeprint:2022/1336}.}

\begin{definition}[OWSGs~\cite{C:MorYam22,cryptoeprint:2022/1336}]
A \emph{one-way state generator (OWSG)} is a triple $(\KeyGen,\StateGen,\Ver)$ of uniform QPT algorithms with the following syntax:
\begin{itemize}
\item $\KeyGen(1^\secp)\to k$:
On input the security parameter~$\secp$, this algorithm outputs a key~$k \in \bit^\lambda$.
\item $\StateGen(k)\to\phi_k$:
On input~$k$, this algorithm outputs the (possibly mixed) quantum state $\phi_k$.
\item $\Ver(k',\phi)\to \{\top, \bot\}$:
On input a bit string $k'$ and a quantum state $\phi$, this algorithm outputs $\top$ or $\bot$.
\end{itemize}
We require the following two properties to hold.
\begin{itemize}
\item (Correctness)
There exists a negligible function $\negl$ such that:
\begin{align*}
   \Pr\left[ \top\gets\Ver(k,\phi_k) \;:\; k\gets\KeyGen(1^\secp),\ \phi_k\gets\StateGen(k) \right] \geq 1-\negl(\secp).
\end{align*}
\item (Security)
For any non-uniform QPT adversary~$\cA$ and any polynomial~$t=t(\lambda)$, there exists a negligible function $\negl$ such that
\begin{align*}
   \Pr\left[ \top\gets\Ver(k',\phi_k) \;:\; k\gets\KeyGen(1^\secp),\ \phi_k\gets\StateGen(k),\ k'\gets\cA(1^\secp,\phi_k^{\otimes t(\lambda)}) \right] \leq \negl(\secp).
\end{align*}
\end{itemize}
\end{definition}

Now we introduce a new variant of OWSGs where we allow $\Ver$ to be inefficient.

\begin{definition}[IV-OWSGs]
An \emph{inefficiently-verifiable one-way state generator (IV-OWSG)} is defined like a OWSG, except that the algorithm~$\Ver$ is allowed to be inefficient (need not be QPT).
\end{definition}

We will also consider exponential security as defined below.
The definition applies to both OWSGs and IV-OWSGs.

\begin{definition}[Exponential security of (IV-)OWSGs]
For a function $\delta \colon \mathbb N \to \mathbb R$,
we say that an OWSG or IV-OWSG is \emph{$\delta$-exponentially secure} if
\begin{align*}
   \Pr\left[ \top\gets\Ver(k',\phi_k) \;:\; k\gets\KeyGen(1^\secp),\ \phi_k\gets\StateGen(k),\ k'\gets\cA(1^\secp,\phi_k^{\otimes t(\lambda)}) \right] \leq 2^{-\delta(\lambda)}
\end{align*}
for any non-uniform QPT adversary $\cA$,  any polynomial $t=t(\secp)$, and a large enough $\secp$, possibly depending on the adversary.
\end{definition}

\section{EFI Pairs Imply IV-OWSGs}
\label{sec:EFI_IVOWSG}

In this section, we show that the existence of EFI pairs implies the existence of IV-OWSGs, unconditionally and without any loss.
For our proof, it will be useful to recall the notion of secretly-verifiable and statistically-invertible one-way state generators (SV-SI-OWSGs) as introduced in~\cite{cryptoeprint:2022/1336}.

\begin{definition}[SV-SI-OWSGs~\cite{cryptoeprint:2022/1336}]
A \emph{secretly-verifiable and statistically-invertible one-way state generator (SV-SI-OWSG)} consists of a pair $(\KeyGen,\StateGen)$ of uniform QPT algorithms with the following syntax.
\begin{itemize}
    \item 
    $\KeyGen(1^\secp)\to k:$
    On input the security parameter $\secp$, this algorithm outputs a key $k \in \bit^\lambda$.
    \item 
    $\StateGen(k)\to\phi_k:$
    On input $k$, this algorithm outputs a (possibly mixed) quantum state $\phi_k$.
\end{itemize}
We require the following two properties to hold:
\begin{itemize}
\item (Statistical Invertibility)
There exists a (not necessarily efficient) POVM~$\{\Pi_k\}_{k\in\bit^\secp}$ such that $\Tr(\Pi_k\phi_k)\ge1-\negl(\secp)$ and $\Tr(\Pi_{k'}\phi_k)\le\negl(\secp)$ for any two keys $k \neq k'$.
\item (Computational Non-Invertibility)
 For any non-uniform QPT adversary $\cA$ and any polynomial $t = t(\lambda)$,  there exists a negligible function $\negl$ such that 
\begin{align*}
   \Pr[k\gets\cA(1^\secp,\phi_k^{\otimes t(\lambda)}):k\gets\KeyGen(1^\secp),\phi_k\gets\StateGen(k)]\le\negl(\secp). 
\end{align*}
\end{itemize}
\end{definition}

The following result was proved in the same paper.

\begin{lemma}[{\cite[Theorem~7.8]{cryptoeprint:2022/1336}}]\label{lem:efi-to-sv-si-owsg}
If EFI pairs exist, then SV-SI-OWSGs exist.
\end{lemma}

We are now ready to state and prove the main theorem of this section.

\begin{theorem}\label{thm1}
If EFI pairs exist, then IV-OWSGs exist.    
\end{theorem}

\begin{proof}
In view of \cref{lem:efi-to-sv-si-owsg} it suffices to show the existence of an IV-OWSG from a SV-SI-OWSG.
Let us thus assume that a SV-SI-OWSG $(\mathsf{SVSI}.\KeyGen,\allowbreak\mathsf{SVSI}.\StateGen)$ exists.
We construct an IV-OWSG $(\KeyGen,\StateGen,\Ver)$ as follows:
\begin{itemize}
\item
$\KeyGen(1^\secp)\to k:$
Run $k\gets\mathsf{SVSI}.\KeyGen(1^\secp)$ and output $k$.
\item 
$\StateGen(k)\to\phi_k:$
Run $\phi_k\gets\mathsf{SVSI}.\StateGen(k)$ and output $\phi_k$.
\item 
$\Ver(k',\phi)\to\top/\bot:$
Measure $\phi$ with the POVM $\{\Pi_k\}_{k\in\bit^\secp}$ that exists by the statistical invertibility.
If the outcome is~$k'$ then output~$\top$.
Otherwise, output $\bot$.
\end{itemize}
Correctness follows as an immediate consequence of the statistical invertibility of the SV-SI-OWSG.

We now prove security.
Let $\cA$ be an arbitrary non-uniform QPT algorithm and $t=t(\lambda)$ a polynomial.
Then it holds that,
\begin{align*}
&\Pr[\top\gets\Ver(k',\phi_k):k\gets\KeyGen(1^\secp),\phi_k\gets\StateGen(k),k'\gets\cA(1^\secp,\phi_k^{\otimes t})]    \\
&=\sum_k\Pr[k\gets\KeyGen(1^\secp)]\sum_{k'} \Pr[k'\gets\cA(1^\secp,\phi_k^{\otimes t})]\cdot\mbox{Tr}(\Pi_{k'}\phi_k)\\
&=\sum_k\Pr[k\gets\KeyGen(1^\secp)]\Pr[k\gets\cA(1^\secp,\phi_k^{\otimes t})]\cdot\mbox{Tr}(\Pi_k\phi_k)\\
&\quad+\sum_k\Pr[k\gets\KeyGen(1^\secp)]\sum_{k'\neq k}\Pr[\alpha\gets\cA(1^\secp,\phi_k^{\otimes t})]\cdot\mbox{Tr}(\Pi_{k'}\phi_k)\\
&\le\sum_k\Pr[k\gets\KeyGen(1^\secp)]\Pr[k\gets\cA(1^\secp,\phi_k^{\otimes t})]\\
&\quad+\sum_k\Pr[k\gets\KeyGen(1^\secp)]\sum_{\alpha\neq k}\Pr[\alpha\gets\cA(1^\secp,\phi_k^{\otimes t})]\cdot\negl(\secp)\\
&\le\negl(\secp),
\end{align*}
where the first inequality follows by upper bounding $\mbox{Tr}(\Pi_k\phi_k) \leq 1$ and using the statistical invertibility of the SV-SI-OWSG, and the second inequality uses the computational non-invertibility of the SV-SI-OWSG.
\end{proof}

\section{Exponentially-Secure IV-OWSGs Imply EFI Pairs}
\label{sec:exp_IVOWSG_EFI}
In this section, we show our main result, namely, the construction of EFI pairs from exponentially-secure IV-OWSGs.
Before stating our main theorem, we first prove a useful lemma.

\begin{lemma}\label{lem:1}
Let $(\KeyGen,\StateGen,\Ver)$ be a $\delta$-exponentially-secure IV-OWSG for some~$\delta=\delta(\secp)$.
Let~$p=p(\secp)$ be an arbitrary polynomial, and define for any~$\secp$ and for any~$k\in\bit^\secp$,
\begin{align*}
G_k\coloneqq G_k(\secp) \coloneqq \left\{k' \in \bit^\secp \;:\; \Pr[\top\gets\Ver(k',\phi_k)]\ge 1-\frac{1}{p(\secp)}\right\}.
\end{align*}
Suppose that $r=r(\secp)$ is a function that satisfies $\delta(\secp) + r(\secp) \geq C \secp$ for some constant~$C>1$.
Then there exists a negligible function~$\negl(\lambda)$ such that the following holds:
\begin{align*}
    \Pr\left[1\le \lvert G_k\rvert \le 2^r \;:\; k\gets\KeyGen(1^\secp) \right]
\geq \Pr\left[ k \in G_k \wedge \lvert G_k\rvert \le 2^r \;:\; k\gets\KeyGen(1^\secp) \right]
\geq 1 - \negl(\lambda),
\end{align*}
\end{lemma}
\begin{proof}
The correctness of the IV-OWSG implies that
\begin{align*}
    \Pr\left[ \lvert G_k\rvert < 1 : k\gets\KeyGen(1^\secp) \right]
\leq 1- \Pr\left[ k \in G_k : k\gets\KeyGen(1^\secp) \right]
\leq \negl(\secp).
\end{align*}
It remains to show that the upper bound $\lvert G_k \rvert \leq 2^r$ holds for all but a negligible fraction of~$k$ as well.
To this end we consider the following trivial attack against the IV-OWSG:\footnote{It works even for $t=0$.}
On input $\phi_k^{\otimes t}$, ignore the state, sample $k' \gets \bit^\secp$, and output $k'$.
By the $\delta$-exponential security of the IV-OWSG, the winning probability of this attack is at most
\begin{align*}
  \Pr\left[ \top\gets\Ver(k',\phi_k) \;:\; k\gets\KeyGen(1^\secp),\ \phi_k\gets\StateGen(k),\ k'\gets \bit^\secp \right]
\leq 2^{-\delta(\secp)}.
\end{align*}
In other words:
\begin{align}\label{eq:trivial succ}
   \sum_{k,k'} \Pr[k] 2^{-\secp} \Pr[ \top\gets\Ver(k',\phi_k) ] \leq 2^{-\delta(\secp)},
\end{align}
where $\Pr[k]\coloneqq\Pr[k\gets\KeyGen(1^\secp)]$.
Let us define the set of keys for which the desired bound is not satisfied
\begin{align*}
  T \coloneqq T(\secp) \coloneqq \left\{k \in \bit^\secp \;:\; \abs{G_k} > 2^{r(\secp)} \right\}.
\end{align*}
Then we can rewrite \cref{eq:trivial succ} as follows:
\begin{align*}
2^{-{\delta(\secp)}}
&\ge\sum_{k,k'}\Pr[k]2^{-\secp} \Pr[\top\gets\Ver(k',\phi_k)]\\
&\ge\sum_{k\in T}\Pr[k] \sum_{k'\in G_k}2^{-\secp}\Pr[\top\gets\Ver(k',\phi_k)]\\
&\ge\sum_{k\in T}\Pr[k] \ \lvert G_k \rvert \ 2^{-\secp} \left( 1 - \frac 1 {p(\secp)} \right)\\
&\ge\sum_{k\in T}\Pr[k] \ 2^{r(\secp)} \ 2^{-\secp} \left( 1 - \frac 1 {p(\secp)} \right).
\end{align*}
Rearranging, we obtain that
\begin{align*}
    \Pr\left[ \lvert G_k\rvert > 2^{r(\secp)} : k\gets\KeyGen(1^\secp) \right]
= \sum_{k\in T}\Pr[k],
\end{align*}
which is negligible since $\delta(\secp) + r(\secp) \geq C \secp$, for some $C>1$.
\end{proof}

We are now ready to show the main theorem of this section.

\begin{theorem}\label{thm:EFI_from_IV-OWSG}
For any constant~$D>0$ the following holds.
If $\delta$-exponentially secure IV-OWSGs exist
with $\delta(\secp)\ge (0.5+D)\secp$,
then EFI pairs exist.
\end{theorem}
\begin{proof}
Assume that an IV-OWSG $(\KeyGen,\StateGen,\Ver)$ exists.
By \cref{cor:com_to_EFI}, it suffices to show that there exists a canonical quantum bit commitment scheme that is statistically $(1-1/\poly(\secp))$-hiding and computationally binding.
To prepare the construction of the bit commitment scheme, we make the following structural observations:
\begin{itemize}
\item
Without loss of generality, we may assume that the $\KeyGen$ algorithm applies a QPT unitary to generate a superposition~$\sum_{k}\sqrt{\Pr[k]} \ket{k} \ket{\mathrm{junk}_k}$, measures the first register, and outputs the result.
\item Similarly, we can assume that the $\StateGen(k)$ algorithm applies a QPT unitary to generate a pure state~$\ket{\Phi_k}_{\regA,\regB}$ and outputs the register~$\regA$, which is in state~$\phi_k=\Tr_{\regB}(\proj{\Phi_k}_{\regA,\regB})$.
\end{itemize}
Let $\cH\coloneqq\{h:\cX\to\cY\}$ be a family of pairwise independent hash functions such that~$\cX\coloneqq\bit^\secp$ and~$\cY\coloneqq\{1,\dots, 2\lfloor 2^{r(\lambda)}\rfloor\}$, where~$r(\secp)\coloneqq(0.5+\frac{D}{2})\secp$.
Let us also denote by~$t(\lambda)$ a polynomial that will be chosen later in the proof.
Then we can define the following two states, which we shall think of being the commitments of $0$ and $1$, respectively:
\begin{align*}
\ket{\psi_{0}}_{\regR,\regC}&\coloneqq
\sum_{k}
\sum_{h\in\cH}
\sqrt{\frac{\Pr[k]}{|\cH|}}
\ket{k,\mathrm{junk}_k,h}_{\regC_1}\ket{h,h(k)}_{\regR_1}\ket{\Phi_k^{\otimes t(\lambda)}}_{\regC_2,\regR_2}\ket{0...0}_{\regR_3}\\
\ket{\psi_1}_{\regR,\regC}&\coloneqq
\sum_{k}
\sum_{h\in\cH}
\sqrt{\frac{\Pr[k]}{|\cH|}}
\ket{k,\mathrm{junk}_k,h}_{\regC_1}\ket{h,h(k)}_{\regR_1}\ket{\Phi_k^{\otimes t(\lambda)}}_{\regC_2,\regR_2}\ket{k}_{\regR_3},
\end{align*}
where $\Pr[k]\coloneqq\Pr[k\gets\KeyGen(1^\secp)]$.
We denote by~$\regC\coloneqq(\regC_1,\regC_2)$ the commitment register and by~$\regR\coloneqq(\regR_1,\regR_2,\regR_3)$ the reveal register.
Note that $\regC_2 \coloneqq (\regB_1,\dots,\regB_t)$ and $\regR_2 \coloneqq (\regA_1,\dots,\regA_t)$, where the $j$-th copy of~$\ket{\Phi_k}$ lives on registers~$\regA_j,\regB_j$.
By construction, both~$\ket{\psi_0}$ and~$\ket{\psi_1}$ can be generated by QPT unitaries.
In the following we show that the scheme is statistically $(1-1/\poly(\secp))$-hiding and computationally binding.

\paragraph{Computational binding.}
Towards a contradiction, we assume that our construction is not computationally binding.
This means that there exists a polynomial $q=q(\lambda)$, an advice quantum state~$\ket{\tau}_\regZ$ on a polynomially-sized register~$\regZ$,  and a non-uniform QPT unitary~$U$ acting on~$\regR$ and~$\regZ$ such that
\begin{align*}
\|\bra{\psi_1}_{\regC,\regR}(I_\regC\otimes U_{\regR,\regZ})(\ket{\psi_0}_{\regC,\regR} \otimes \ket{\tau}_\regZ)\|^2\ge\frac{1}{q(\secp)}
\end{align*}
for infinitely many~$\secp$.
This means that, abbreviating $t=t(\secp)$,
\begin{align}
\frac{1}{q(\secp)}
&\le\left\|
\sum_{k,h}\frac{\Pr[k]}{|\cH|}
\bra{h,h(k)}_{\regR_1}\bra{\Phi_k^{\otimes t}}_{\regC_2,\regR_2}\bra{k}_{\regR_3}
(I_{\regC_2}\otimes U_{\regR,\regZ})
\ket{h,h(k)}_{\regR_1}\ket{\Phi_k^{\otimes t}}_{\regC_2,\regR_2}\ket{0...0}_{\regR_3}\ket{\tau}_\regZ\right\|^2\nonumber\\
&\le\left(\sum_{k,h}\frac{\Pr[k]}{|\cH|}
\left\|\bra{h,h(k)}_{\regR_1}\bra{\Phi_k^{\otimes t}}_{\regC_2,\regR_2}\bra{k}_{\regR_3}
(I_{\regC_2}\otimes U_{\regR,\regZ})
\ket{h,h(k)}_{\regR_1}\ket{\Phi_k^{\otimes t}}_{\regC_2,\regR_2}\ket{0...0}_{\regR_3}\ket{\tau}_\regZ\right\|\right)^2\nonumber\\
&\le\sum_{k,h}\frac{\Pr[k]}{|\cH|}
\left\|\bra{h,h(k)}_{\regR_1}\bra{\Phi_k^{\otimes t}}_{\regC_2,\regR_2}\bra{k}_{\regR_3}
(I_{\regC_2}\otimes U_{\regR,\regZ})
\ket{h,h(k)}_{\regR_1}\ket{\Phi_k^{\otimes t}}_{\regC_2,\regR_2}\ket{0...0}_{\regR_3}\ket{\tau}_\regZ\right\|^2\nonumber\\
&\le\sum_{k,h}\frac{\Pr[k]}{|\cH|}
\left\|\bra{k}_{\regR_3}
(I_{\regC_2}\otimes U_{\regR,\regZ})
\ket{h,h(k)}_{\regR_1}\ket{\Phi_k^{\otimes t}}_{\regC_2,\regR_2}\ket{0...0}_{\regR_3}\ket{\tau}_\regZ\right\|^2\nonumber\\
&\le\sum_{k\in G,h}\frac{\Pr[k]}{|\cH|}
\left\|\bra{k}_{\regR_3}
(I_{\regC_2}\otimes U_{\regR,\regZ})
\ket{h,h(k)}_{\regR_1}\ket{\Phi_k^{\otimes t}}_{\regC_2,\regR_2}\ket{0...0}_{\regR_3}\ket{\tau}_\regZ\right\|^2\nonumber\\
&+\sum_{k\notin G,h}\frac{\Pr[k]}{|\cH|}
\left\|\bra{k}_{\regR_3}
(I_{\regC_2}\otimes U_{\regR,\regZ})
\ket{h,h(k)}_{\regR_1}\ket{\Phi_k^{\otimes t}}_{\regC_2,\regR_2}\ket{0...0}_{\regR_3}\ket{\tau}_\regZ\right\|^2\nonumber\\
&\le\sum_{k\in G,h}\frac{\Pr[k]}{|\cH|}
\left\|\bra{k}_{\regR_3}
(I_{\regC_2}\otimes U_{\regR,\regZ})
\ket{h,h(k)}_{\regR_1}\ket{\Phi_k^{\otimes t}}_{\regC_2,\regR_2}\ket{0...0}_{\regR_3}\ket{\tau}_\regZ\right\|^2+\negl(\secp).\label{from_nobinding}
\end{align}
Here, in the second inequality we have used the triangle inequality, and in the third inequality we have used Jensen's inequality.
$G$ is a set defined as
\begin{align}
G\coloneqq\left\{k\in\bit^\secp:\Pr[\top\gets\Ver(k,\phi_k)]\ge\frac{1}{2}\right\}.
\end{align}
Due to the correctness of the IV-OWSG, we have $\sum_{k\not\in G}\Pr[k]\le\negl(\secp)$.

From such $\ket{\tau}$ and $U$, we can construct a non-uniform QPT adversary~$\cA$ that breaks the security of the IV-OWSG as follows:
\begin{enumerate}
    \item
    The algorithm gets as input the $\regR_2$ register of $\ket{\Phi_k^{\otimes t}}_{\regC_2,\regR_2}$ and the $\regZ$ register containing the advice state $\ket{\tau}_\regZ$.
    \item
    Choose $h\gets\cH$. Choose $y\gets \cY$.
    \item
    Apply $(I_{\regC_2}\otimes U_{\regR,\regZ})$ on $\ket{h,y}_{\regR_1}\ket{\Phi_k^{\otimes t}}_{\regC_2,\regR_2}\ket{0...0}_{\regR_3}\ket{\tau}_\regZ$.
    \item
    Measure the register $\regR_3$ in the computational basis, and output the result.
\end{enumerate}
The probability that $\cA$ outputs $k$ and $k\in G$ can then be lower bounded as
\begin{align*}
&\Pr\left[ k\in G\wedge k\gets\cA(1^\secp,\phi_k^{\otimes t(\secp)}) \;:\; k\gets\KeyGen(1^\secp),\ \phi_k\gets\StateGen(k) \right] \\
&= \sum_{k\in G} \Pr[k] \sum_{h,y} \frac1{|\cH| \cdot|\cY|}
\left\|\bra{k}_{\regR_3}
(I_{\regC_2}\otimes U_{\regR,\regZ})
\ket{h,y}_{\regR_1}\ket{\Phi_k^{\otimes t}}_{\regC_2,\regR_2}\ket{0...0}_{\regR_3}\ket{\tau}_\regZ\right\|^2\\
&= \sum_{k\in G,h} \frac{\Pr[k]}{|\cH|} \sum_y \frac{1}{2\lfloor 2^{r(\secp)}\rfloor}
\left\|\bra{k}_{\regR_3}
(I_{\regC_2}\otimes U_{\regR,\regZ})
\ket{h,y}_{\regR_1}\ket{\Phi_k^{\otimes t}}_{\regC_2,\regR_2}\ket{0...0}_{\regR_3}\ket{\tau}_\regZ\right\|^2\\
&\ge
\sum_{k\in G,h}\frac{\Pr[k]}{|\cH|}
\frac{1}{2\lfloor 2^{r(\secp)}\rfloor}
\left\|\bra{k}_{\regR_3}
(I_{\regC_3}\otimes U_{\regR,\regZ})
\ket{h,h(k)}_{\regR_1}\ket{\Phi_k^{\otimes t}}_{\regC_2,\regR_2}\ket{0...0}_{\regR_3}\ket{\tau}_\regZ\right\|^2\\
&\ge\frac{1}{4q(\secp)\lfloor 2^{r(\secp)}\rfloor},
\end{align*}
where the first inequality is obtained by only keeping the term~$y=h(k)$,
and the last inequality is from \cref{from_nobinding}.
Hence, the probability that $\cA$ wins is lowerbounded by
$\frac{1}{8q(\secp)\lfloor 2^{r(\secp)}\rfloor}$.

We claim that this is larger than~$2^{-\delta(\secp)}$, in contradiction to the $\delta$-exponential security of the IV-OWSG.
Indeed, we have
\begin{align*}
\frac{1}{8q(\secp)\lfloor 2^{r(\secp)}\rfloor}-\frac{1}{2^{\delta(\secp)}}
=\frac{2^{\delta(\secp)}-8q(\secp)\lfloor 2^{r(\secp)}\rfloor}{8q(\secp)2^{\delta(\secp)}\lfloor 2^{r(\secp)}\rfloor}
\ge\frac{2^{\delta(\secp)}-8q(\secp)2^{r(\secp)}}{8q(\secp)2^{\delta(\secp)}\lfloor 2^{r(\secp)}\rfloor}
\ge\frac{2^{(0.5 + D) \secp} - 8q(\secp)2^{(0.5 + \frac D 2) \secp}}{8q(\secp)2^{\delta(\secp)}\lfloor 2^{r(\secp)}\rfloor}
>0
\end{align*}
for infinitely many~$\secp$.
Thus we have obtained the desired contradiction, and we may conclude that our bit commitment scheme is computationally binding.

\paragraph{Statistical hiding.}
We first describe an (inefficient) algorithm that, given~$\ket{h,h(k)}_{\regR_1}$ and the~$\regR_2$ register of~$\ket{\Phi_k^{\otimes t(\secp)}}_{\regC_2,\regR_2}$, for~$k \leftarrow \KeyGen(1^\secp)$ and $h \leftarrow \cH$, outputs the key~$k$ with not too small probability:
\begin{enumerate}
\item Apply the shadow tomography procedure on the~$\regR_2$ register of $\ket{\Phi_k^{\otimes t(\secp)}}_{\regC_2,\regR_2}$ to obtain an estimate~$b_{k'}$ of $\Pr[\top\gets\Ver(k',\phi_k)]$ with additive error
$\varepsilon = 1/8$
for all~$k'\in\bit^\secp$, with failure probability at most~$\omega = 2^{-\secp}$.
Compute the list
$\mathcal L \coloneqq \{ k'\in\bit^\secp : b_{k'} \geq 3/4 \}$.
\item
Read $(h,h(k))$ from $\regR_1$.
Compute $h(k')$ for each $k'$ in the list~$\mathcal L$.
If there is only a single $k^*$ in the list such that $h(k^*)=h(k)$, output $k^*$.
Otherwise, output $\bot$.
\end{enumerate}
By \cref{lem:ST}, the first step is possible if we choose~$t(\secp)$ to be a large enough polynomial.
Then, with probability at least~$1-\omega=1-\negl(\lambda)$, the list~$\mathcal L$ computed in the first step satisfies
\begin{align}
\left\{ k' \in \bit^\secp : \Pr[\top\gets\Ver(k',\phi_k)] \geq \frac{7}{8} \right\}
\subseteq \mathcal L
\subseteq G_k,
\end{align}
where $G_k =  \left\{ k' \in \bit^\secp : \Pr[\top\gets\Ver(k',\phi_k)] \geq \frac{1}{2} \right\}$ is the set defined in \cref{lem:1}.
By the correctness of the IV-OWSG, it holds that, except for a negligible fraction of $k$, $\Pr[\top\gets\Ver(k,\phi_k)]\geq \frac{7}{8}$ holds for sufficiently large~$\secp$, in which case the first inclusion implies that~$k \in \mathcal L$.
This means that we have
\begin{align}\label{eq:step 1}
  \Pr\left[ k \in \mathcal L \subseteq G_k : k \leftarrow \KeyGen(1^\secp) \right] \geq 1 - \negl(\secp).
\end{align}
We claim that the second step will return~$k$ with not too small probability.
To see this, we first note that for any fixed set~$G \subseteq \bit^\secp$ and for any fixed~$x \in G$, it holds that
\begin{align*}
  \Pr_{h \leftarrow \cH}\left[ \lvert G \cap h^{-1}(h(x)) \rvert = 1 \right]
&= 1 - \Pr_{h \leftarrow \cH}\left[ \bigvee_{x' \in G, x' \neq x} h(x') = h(x) \right] \\
&\geq 1 - \sum_{x' \in G, x' \neq x} \Pr_{h \leftarrow \cH}\left[ h(x') = h(x) \right]
= 1 - \frac{|G|-1}{|\cY|}
\end{align*}
due to the pairwise independence of the family of hash functions and the union bound.
Thus,
\begin{align*}
&\quad \Pr\left[ \mathcal L \cap h^{-1}(h(k)) = \{k\} : k \leftarrow \KeyGen(1^\secp), \ h \leftarrow \cH \right] \\
&= \Pr\left[ \lvert \mathcal L \cap h^{-1}(h(k)) \rvert = 1,\ k \in \mathcal L \;:\; k \leftarrow \KeyGen(1^\secp), \ h \leftarrow \cH \right] \\
&\geq \Pr\left[ \lvert G_k \cap h^{-1}(h(k)) \rvert = 1,\ k \in \mathcal L \subseteq G_k \;:\; k \leftarrow \KeyGen(1^\secp), \ h \leftarrow \cH \right] \\
&\geq \Pr_{h \leftarrow \cH}\left[ \lvert G_k \cap h^{-1}(h(k)) \rvert = 1, \ k \in G_k \;:\; k \leftarrow \KeyGen(1^\secp) \right] \left( 1 - \negl(\lambda) \right) \\
&\geq \Pr_{h \leftarrow \cH}\left[ \lvert G_k \cap h^{-1}(h(k)) \rvert = 1, \ k \in G_k, \ |G_k| \leq 2^{r(\secp)} \;:\; k \leftarrow \KeyGen(1^\secp) \right] \left( 1 - \negl(\lambda) \right) \\
&\geq \Pr_{h \leftarrow \cH}\left[ \lvert G_k \cap h^{-1}(h(k)) \rvert = 1 \;:\; k \leftarrow \KeyGen(1^\secp), \ k \in G_k, \ |G_k| \leq 2^{r(\secp)} \right] \left( 1 - \negl(\lambda) \right)^2 \\
&\geq \left( 1-\frac{2^{r(\secp)}-1}{2\lfloor 2^{r(\secp)}\rfloor} \right) \left( 1 - \negl(\lambda) \right)
\geq \frac{1}{2} - \negl(\lambda),
\end{align*}
where the second inequality uses \cref{eq:step 1} and the fact that~$h$ and~$\mathcal L$ are independent conditioned on~$k$, the fourth inequality uses \cref{lem:1}, and the last line follows from the bound derived above.
We conclude that the algorithm returns the correct key with probability at least~$\frac12 - \negl(\secp)$.

We may summarize the above as saying that there exists an (inefficient) POVM measurement~$\{ \Pi_{\regR_1,\regR_2}^{(\alpha)} \}_{\alpha \in \bit^\secp \cup \{\bot\}}$ such that
\begin{equation}\label{eq:hiding summary}
    \sum_k \sum_{h \in \cH} \frac{\Pr[k]}{|\cH|}
\bra{h,h(k)}_{\regR_1} \bra{\Phi_k^{\otimes t(\lambda)}}_{\regC_2,\regR_2}
\Pi_{\regR_1,\regR_2}^{(k)}
\ket{h,h(k)}_{\regR_1} \ket{\Phi_k^{\otimes t(\lambda)}}_{\regC_2,\regR_2} \geq \frac12 - \negl(\secp).
\end{equation}
This implies that there is an (inefficient) unitary~$U_{\regG,\regZ}$ that sends~$\ket{\psi_0}_{\regR,\regC} \ket0_{\regZ}$ to a state with not too small overlap with~$\ket{\psi_1}_{\regR,\regC} \ket0_{\regZ}$.
To see this, choose a unitary~$V_{\regR_1,\regR_2,\regZ}$ that extends the Naimark dilation~$V_{\regR_1,\regR_2,\regZ} \ket0_{\regZ} = \sum_\alpha (\Pi_{\regR_1,\regR_2}^{(\alpha)})^{1/2} \otimes \ket\alpha_{\regZ}$ of the POVM measurement, and define~$U_{\regG,\regZ} \coloneqq V_{\regR_1,\regR_2,\regZ}^\dagger {\mathrm{CNOT}}_{\regZ\to\regR_3} V_{\regR_1,\regR_2,\regZ}$.
Then it holds that
\begin{align*}
&\quad \bra{\psi_1}_{\regR,\regC} \bra0_{\regZ} U_{\regR,\regZ} \ket{\psi_0}_{\regR,\regC} \ket0_{\regZ} \\
&= \sum_k \sum_{h \in \cH} \frac{\Pr[k]}{|\cH|} \bra{h,h(k)}_{\regR_1} \bra{\Phi_k^{\otimes t(\lambda)}}_{\regC_2,\regR_2} \bra{k}_{\regR_3} \bra{0}_{\regZ} U_{\regR,\regZ} \ket{h,h(k)}_{\regR_1} \ket{\Phi_k^{\otimes t(\lambda)}}_{\regC_2,\regR_2} \ket{0}_{\regR_3} \ket{0}_{\regZ} \\
&= \sum_k \sum_{h \in \cH} \frac{\Pr[k]}{|\cH|} \sum_{\alpha,\beta} \bra{h,h(k)}_{\regR_1} \bra{\Phi_k^{\otimes t(\lambda)}}_{\regC_2,\regR_2}
\sqrt{\Pi_{\regR_1,\regR_2}^{(\alpha)}}
\bra{\alpha,k}_{\regZ,\regR_3} {\mathrm{CNOT}}_{\regZ\to\regR_3} \ket{\beta,0}_{\regZ,\regR_3} \\
&\qquad\qquad\qquad\qquad\quad \sqrt{\Pi_{\regR_1,\regR_2}^{(\beta)}}
\ket{h,h(k)}_{\regR_1} \ket{\Phi_k^{\otimes t(\lambda)}}_{\regC_2,\regR_2} \\
&= \sum_k \sum_{h \in \cH} \frac{\Pr[k]}{|\cH|}
\bra{h,h(k)}_{\regR_1} \bra{\Phi_k^{\otimes t(\lambda)}}_{\regC_2,\regR_2}
\Pi_{\regR_1,\regR_2}^{(k)}
\ket{h,h(k)}_{\regR_1} \ket{\Phi_k^{\otimes t(\lambda)}}_{\regC_2,\regR_2} \\
&\geq \frac12 - \negl(\secp),
\end{align*}
where the last step is due to \cref{eq:hiding summary}.
Hence
\begin{align*}
  \left| \bra{\psi_1}_{\regR,\regC} \bra0_{\regZ} U_{\regR,\regZ} \ket{\psi_1}_{\regR,\regC} \ket0_{\regZ} \right|^2
\geq \frac14 - \negl(\secp).
\end{align*}
By Uhlmann's theorem, it follows that the reduced states $\rho_b\coloneqq\Tr_{\regR}(\proj{\psi_b}_{\regR,\regC})$ have fidelity
\[
F(\rho_0,\rho_1) \geq \frac14 - \negl(\lambda),
\]
and hence their trace distance is at most
\[
  \TD(\rho_0,\rho_1)
\leq \sqrt{\frac34} + \negl(\lambda).
\]
In view of \cref{eq:stat hiding via TD} this implies that the scheme is statistically $(\sqrt{\frac34} + \negl(\lambda))$-hiding, which in particular means that it is statistically $(1-1/\poly(\secp))$-hiding.
This completes the proof of \Cref{thm:EFI_from_IV-OWSG}. 
\end{proof}

\subsection*{Acknowledgements}

GM~was supported by the European Research Council through an ERC Starting Grant (Grant agreement No.~101077455, ObfusQation).
GM and MW acknowledge support by the Deutsche Forschungsgemeinschaft (DFG, German Research Foundation) under Germany's Excellence Strategy - EXC\ 2092\ CASA - 390781972.
MW also acknowledges support by the European Research Council through an ERC Starting Grant (grant agreement No.~101040907, SYMOPTIC), by the NWO through grant OCENW.KLEIN.267, and by the BMBF through project Quantum Methods and Benchmarks for Resource Allocation (QuBRA).
TM is supported by JST CREST JPMJCR23I3, JST Moonshot JPMJMS2061-5-1-1, JST FOREST, MEXT QLEAP, the Grant-in Aid for Transformative Research Areas (A) 21H05183, and the Grant-in-Aid for Scientific Research (A) No.22H00522.

\ifnum\submission=0
\bibliographystyle{alpha}
\else
\bibliographystyle{splncs04}
\fi
\bibliography{abbrev3,crypto,reference}

\end{document}